\documentclass[10pt,twocolumn]{article}
\usepackage[margin=0.67in,columnsep=0.24in]{geometry}
\usepackage{amsmath,amssymb,mathtools,amsthm,bm}
\usepackage{microtype}
\usepackage{booktabs,array,tabularx}
\usepackage{graphicx}
\usepackage{tikz}
\usetikzlibrary{arrows.meta,positioning,calc,decorations.pathreplacing}
\usepackage{xcolor}
\usepackage{enumitem}
\usepackage{hyperref}
\hypersetup{colorlinks=true,linkcolor=blue!45!black,citecolor=blue!45!black,urlcolor=blue!45!black}
\usepackage{titlesec}
\usepackage{authblk}
\usepackage{balance}

\setlist{nosep,leftmargin=*}
\titleformat{\section}{\large\bfseries}{\thesection}{0.55em}{}
\titleformat{\subsection}{\normalsize\bfseries}{\thesubsection}{0.5em}{}
\titlespacing*{\section}{0pt}{1.2ex plus .4ex minus .2ex}{.55ex}
\titlespacing*{\subsection}{0pt}{1.0ex plus .3ex minus .2ex}{.4ex}

\newtheorem{theorem}{Theorem}

\newtheorem{proposition}[theorem]{Proposition}
\theoremstyle{remark}

\newcommand{\Aalg}{\mathfrak A}
\newcommand{\Nalg}{\mathfrak N}
\newcommand{\RR}{\mathbb R}
\newcommand{\ZZ}{\mathbb Z}
\newcommand{\id}{\mathrm{id}}

\newcommand{\sllim}{\mathop{\mathrm{s\!\mbox{-}\!lim}}\limits}

\title{\bfseries Searching for Bell-CHSH-Violating Bose Operators}
\author[1]{Tony Newton}
\author[2]{Silvio Paolo Sorella}
\affil[1]{Newton Astro Labs, London, UK\\\texttt{tony.newton79@gmail.com}}
\affil[2]{UERJ -- Universidade do Estado do Rio de Janeiro, Instituto de F\'isica -- Departamento de F\'isica Te\'orica, Rio de Janeiro, Brazil\\\texttt{silvio.sorella@gmail.com}}
\date{}

\begin{document}
\twocolumn[
\begin{@twocolumnfalse}
\maketitle
\vspace{-1.5em}
\begin{abstract}
We construct bounded Bell observables for a free scalar field in complementary spacetime wedges using measurable involutions of the spectral variable of a local canonical quadrature. These involutions yield wedge-local Hermitian unitary operators that anticommute with a chosen bounded observable. For a cosine observable, we characterize the admissible translations and reflections and derive exact correlations at finite modular bandwidth. The affine-reflection construction yields a violation of the Bell--CHSH inequality of approximately 2.06.
\end{abstract}
\vspace{0.8em}
\end{@twocolumnfalse}
]

\section{Introduction}
Vacuum Bell correlations in relativistic quantum field theory are known to violate the Bell--CHSH inequality, and in broad algebraic settings the Tsirelson bound can be attained \cite{Summers1985,SummersWerner1987a,SummersWerner1987b,SummersWerner1987c}. The constructive problem is more specific: given a local canonical pair in complementary wedges, can one write bounded local observables whose algebra and finite-band correlations can be analyzed exactly?

Recent work has produced several explicit answers using Weyl operators, unitary deformations, modular localization, coherent-state interference, vertex operators and finite Weyl polynomials \cite{DeFabritiis2023,Caribe2026Modular,Guimaraes2026Cat,Azevedo2026,Caribe2026Vertex,Caribe2026FiniteWeyl}. In particular, finite Weyl polynomials can approach the Tsirelson bound, while chiral-boson vertex operators provide a saturating construction. The purpose of the present paper is therefore not to improve the largest known value of the Bell correlator. We instead isolate a different structural mechanism: exact anticommutation realized by permutations of spectral fibres of a local Bose quadrature.

There is an important continuous-variable antecedent. Position-space pseudospins may be built from interval projections and box translations; Larsson's construction \cite{Larsson2004} is the adjacent-cell member of the translation family considered below. We do not claim that interval swaps or pseudospin algebras are new. The additional ingredients here are their embedding into a modularly localized CCR plane, a complete classification of the composition-operator solutions for the cosine axis, and an exact finite-band reduction of the two-wedge correlation problem.

The mathematical structure separates into three layers. First, the anticommutation equation becomes a measurable matching problem between opposite spectral fibres. Second, Stone--von Neumann uniqueness with multiplicity identifies the von Neumann algebra generated by a nondegenerate local CCR plane with a type-I factor acting on one Schr\"odinger degree of freedom, so the fibre involutions give genuine wedge-local operators. Third, quasifree correlations reduce the translation-cell problem to a positive Gaussian quadratic assignment functional. This last reduction yields both a rearrangement theorem inside one length sector and a counterexample to a stronger universal shell-compression conjecture. The reflection branch admits a closed formula and already gives an analytic finite-band Bell violation without any combinatorial search.

The operator-theoretic ingredients used in the locality step are standard. The contribution is the way they combine with modular localization and the fibre geometry to give an explicit, exactly computable class of wedge-local Bell observables. Numerical enumeration is used only to illuminate the residual finite combinatorics and is moved to the appendices wherever it is not needed for the main statements.

\section{Modular canonical pair and fibre classification}
Let $\varphi(h)$ be a real smeared free scalar field and $W(h)=e^{i\varphi(h)}$ the Weyl operator, with
\begin{align}
W(h)W(k)&=e^{-i\Delta_{\mathrm{PJ}}(h,k)/2}W(h+k),\\
\Delta_{\mathrm{PJ}}(h,k)&=2\operatorname{Im}\langle h|k\rangle .
\end{align}
Choose a symplectically normalized local pair $(f,f')$ such that $\Delta_{\mathrm{PJ}}(f,f')=2$, and set
\begin{equation}
\alpha=\sqrt{\frac{\pi}{2}},\qquad A=\alpha\varphi(f),\qquad A'=\alpha\varphi(f').
\end{equation}
Then $[A,A']=i\pi$. Writing
\begin{equation}
q=\frac{A}{\sqrt\pi},\qquad p=\frac{A'}{\sqrt\pi},\qquad [q,p]=i,
\end{equation}
we set $a=\sqrt\pi$ so that $A=aq$.

The finite-band modular packet is encoded by three positive covariance parameters $\nu_A,\nu_B,\kappa$. With
\begin{equation}
\lambda^2=e^{-2\pi\omega_0},\qquad d=r\omega_0,\qquad S_d=\frac{\sinh(2\pi d)}{2\pi d},
\end{equation}
let $\mu_+=\lambda^2S_d$ and $\mu_-=\lambda^{-2}S_d$. Then
\begin{align}
\nu_A&=\frac{2}{1-\mu_+}-1,\qquad
\nu_B=\frac{2}{\mu_--1}+1,\\
\kappa&=\frac{2}{\sqrt{(1-\mu_+)(\mu_--1)}}.
\end{align}
For $0<r<1$ and $\omega_0\downarrow0$, all three grow as $(\pi\omega_0)^{-1}+O(1)$ while $\nu_A+\nu_B-2\kappa\to0$.

Consider the composition operator
\begin{equation}
(\Gamma_T\psi)(q)=\psi(Tq)
\end{equation}
associated with a measurable measure-preserving involution $T$. Exact anticommutation with $X=\cos(aq)$ is equivalent to
\begin{equation}\label{eq:fibrecondition}
\cos(aTq)=-\cos(aq)
\end{equation}
almost everywhere.

\begin{theorem}[Cosine-fibre classification]\label{thm:classification}
Every pointwise solution of \eqref{eq:fibrecondition} has one of the forms
\begin{equation}
T(q)=q+(2k+1)a,
\end{equation}
or
\begin{equation}
T(q)=-q+(2k+1)a,
\end{equation}
where $k\in\ZZ$ may depend measurably on $q$. If $T^2=\id$, the branch labels are constrained so that $T$ is an involutive matching between opposite fibres of the cosine spectrum.
\end{theorem}

\begin{proof}
Equation \eqref{eq:fibrecondition} is equivalent to $\cos(aTq)=\cos(\pi-aq)$. Hence $aTq=2\pi n\pm(\pi-aq)$, which gives the two displayed branches. The involution condition supplies the consistency rule on paired points.
\end{proof}

For almost every $x\in(0,\pi)$ the fibre of $\cos$ at $\cos x$ is
\begin{equation}
\mathcal F_x=\{2\pi n+x,2\pi n-x:n\in\ZZ\},
\end{equation}
while the opposite fibre is $\mathcal F_{\pi-x}$. Exact Cliffordization is therefore a measurable perfect matching $\mathcal F_x\leftrightarrow\mathcal F_{\pi-x}$.

\begin{figure}[t]
\centering
\begin{tikzpicture}[x=0.44cm,y=0.55cm,>=Latex,thick]
\draw[->] (-5.5,0)--(5.8,0) node[right] {$q$};
\foreach \x in {-4,-2,0,2,4}{\draw (\x,-0.12)--(\x,0.12);}
\node[below] at (-4,0) {$-2a$}; \node[below] at (-2,0) {$-a$};
\node[below] at (0,0) {$0$}; \node[below] at (2,0) {$a$}; \node[below] at (4,0) {$2a$};
\draw[blue!60!black,->,bend left=25] (-1.35,0.55) to (0.65,0.55);
\node[blue!60!black,above] at (-0.35,1.05) {odd translation};
\draw[red!65!black,->,bend left=35] (2.9,-0.55) to (-0.9,-0.55);
\node[red!65!black,below] at (1.05,-1.05) {odd-centred reflection};
\end{tikzpicture}
\caption{Schematic of the two pointwise branches in Theorem~\ref{thm:classification}. The drawing is explanatory; the measurable branch label may vary with $q$ subject to the involution constraint.}
\label{fig:fibres}
\end{figure}

\section{Wedge-local realization}
The fibre description becomes a QFT observable only after algebra membership is established. Let $K_R$ denote the standard real one-particle subspace of the right wedge and let $(f,f')\subset K_R$ be the normalized pair above. Define
\begin{equation}
E_R=\operatorname{span}_{\RR}\{f,f'\},\qquad
\Nalg_R=\{W(h):h\in E_R\}^{\prime\prime}.
\end{equation}
By isotony, $\Nalg_R\subset\Aalg(W_R)$. Since the symplectic form is nondegenerate on $E_R$, the regular representation of this finite-dimensional CCR is, by Stone--von Neumann uniqueness with multiplicity, unitarily equivalent to one Schr\"odinger degree of freedom tensored with a multiplicity space \cite{Takesaki,BratteliRobinson}. Thus a unitary $U_R$ exists such that
\begin{equation}\label{eq:typeI}
U_R\Nalg_RU_R^{-1}=B(L^2(\RR))\otimes I,
\end{equation}
and $q,p$ act canonically on the first factor.

\begin{theorem}[Local fibre-Cliffordization in a CCR plane]\label{thm:wedge}
Let $F\in L^\infty(\RR)$ be real-valued and let $T:\RR\to\RR$ be a measurable, Lebesgue-measure preserving involution. Then $C_T\psi=\psi\circ T$ is a Hermitian unitary on $L^2(\RR)$ and
\begin{equation}\label{eq:lift}
\Gamma_T=U_R^{-1}(C_T\otimes I)U_R\in\Aalg(W_R).
\end{equation}
Moreover
\begin{equation}\label{eq:generalconj}
\Gamma_T F(q)\Gamma_T=F(Tq).
\end{equation}
Consequently $\{\Gamma_T,F(q)\}=0$ if and only if $F(Tq)=-F(q)$ almost everywhere. The corresponding construction in the left wedge produces $\Gamma_B\in\Aalg(W_L)$, and wedge locality gives $[\Gamma_A,\Gamma_B]=0$.
\end{theorem}

\begin{proof}
Measure preservation makes $C_T$ unitary and $T^2=\id$ makes it Hermitian. Equation \eqref{eq:typeI} places $C_T\otimes I$ in the local CCR factor. In the Schr\"odinger representation one has $C_TM_FC_T=M_{F\circ T}$ for every bounded Borel function $F$, which proves \eqref{eq:generalconj}; the anticommutation criterion follows immediately. The left-wedge statement is identical.
\end{proof}

The theorem is a local realization principle rather than a new classification theorem in abstract operator theory. In general, the existence of a Hermitian unitary anticommuting with a self-adjoint operator is governed by the unitary equivalence of $X$ and $-X$, hence by the symmetry of their spectral multiplicities; see, for example, \cite{Pedersen1990}. What is useful here is that the nondegenerate free-field CCR plane supplies a concrete type-I subfactor inside the wedge algebra, while the composition-operator subclass is explicit enough for the correlations to be calculated. No analogous canonical type-I statement is asserted for arbitrary interacting local algebras.

For translation-cell involutions there is a more explicit local construction. Partition $\RR$ into
\begin{equation}
I_n=[na,(n+1)a),\qquad n\in\ZZ,
\end{equation}
and let $\sigma:\ZZ\to\ZZ$ be a fixed-point-free involution such that $\sigma(n)-n$ is odd. For
\begin{equation}
E_n=\mathbf 1_{I_n}(q),\qquad d_n=[\sigma(n)-n]a,\qquad V_n=e^{id_np},
\end{equation}
the partial sums over finite $\sigma$-invariant sets $F$,
\begin{equation}
\Gamma_{\sigma,F}=\sum_{n\in F}E_nV_n,
\end{equation}
converge strongly to
\begin{equation}\label{eq:strongsum}
\Gamma_\sigma=\sllim_{F\uparrow\ZZ}\sum_{n\in F}E_nV_n\in\Aalg(W_R).
\end{equation}
Indeed, orthogonality of the spectral ranges gives
\begin{equation}
\left\|\sum_{n\notin F}E_nV_n\psi\right\|^2
=\sum_{n\notin F}\|E_{\sigma(n)}\psi\|^2\longrightarrow0.
\end{equation}
Complete transposition pairs are adjoints of one another, so the limit is a Hermitian unitary; the odd displacement implies $\{\Gamma_\sigma,\cos A\}=0$ exactly.

The adjacent pairing $\sigma(2n)=2n+1$, $\sigma(2n+1)=2n$ is, after rescaling the interval length, the box-translation pseudospin of Larsson \cite{Larsson2004}. Thus the translation-cell construction should be read as a wedge-local extension and re-pairing of an established continuous-variable mechanism, not as a claim that interval pseudospins themselves are new. The freedom to use arbitrary odd matchings, together with the reflection branch of Theorem~\ref{thm:classification}, is what leads to the finite-band geometry studied below.

\section{Exact Gaussian functional for translation cells}
A transposition $\{n,m\}$ is described by its midpoint and odd length,
\begin{equation}
c=\frac{n+m+1}{2}\in\ZZ,
\qquad \ell=|m-n|\in\{1,3,5,\ldots\}.
\end{equation}
Let
\begin{equation}
J_c=[(c-\tfrac12)a,(c+\tfrac12)a).
\end{equation}
For the complementary wedges define
\begin{equation}
Q=\frac12\begin{pmatrix}\nu_A&\kappa\\\kappa&\nu_B\end{pmatrix},\qquad
P=\frac12\begin{pmatrix}\nu_A&-\kappa\\-\kappa&\nu_B\end{pmatrix}.
\end{equation}
If $M\sim N(0,Q)$, set
\begin{equation}
P_{cc'}=\mathbb P(M_A\in J_c,M_B\in J_{c'}).
\end{equation}
For positive odd lengths $\ell,m$, summing both translation orientations gives
\begin{align}
W_{\ell m}={}&2\exp\!\left[-\frac{\pi}{4}(\nu_A\ell^2+\nu_Bm^2-2\kappa\ell m)\right]\notag\\
&+2\exp\!\left[-\frac{\pi}{4}(\nu_A\ell^2+\nu_Bm^2+2\kappa\ell m)\right].
\end{align}

\begin{theorem}[Finite-band translation functional]\label{thm:qap}
Let $\mathcal P(\sigma)$ denote the transpositions of a finite or probability-truncated translation matching, with midpoint-length data $(c_p,\ell_p)$. Then
\begin{equation}\label{eq:qap}
C_\Gamma(\sigma)=\sum_{p,q\in\mathcal P(\sigma)}P_{c_pc_q}W_{\ell_p\ell_q}.
\end{equation}
This identity is exact at nonzero modular bandwidth.
\end{theorem}

\begin{proof}
For a cell translated by signed odd length $s\ell$, the quasifree Gaussian kernel separates into a midpoint integral and a displacement weight. The midpoint integration over the paired cells gives $P_{cc'}$; summing the four orientation pairs $(s,t)\in\{\pm1\}^2$ gives $W_{\ell m}$. Summation over transpositions yields \eqref{eq:qap}.
\end{proof}

The functional is a structured quadratic assignment problem \cite{Cela1998,Burkard2009,BurkardKlinz1996,Laurent2015}. Its useful feature is a strong separation between distinct length sectors. Define
\begin{equation}
D=\nu_A\nu_B-\kappa^2>0,
\qquad \delta=\nu_A+\nu_B-2\kappa>0,
\qquad \gamma=\frac{D}{\delta}.
\end{equation}
A constrained quadratic minimization gives
\begin{equation}
\nu_Ax^2+\nu_By^2-2\kappa xy\ge\gamma(x-y)^2
\end{equation}
and the corresponding $+$ inequality with $(x+y)^2$. Hence distinct odd lengths obey
\begin{equation}\label{eq:offlength}
W_{\ell m}\le 4e^{-\pi\gamma},\qquad \ell\ne m.
\end{equation}
At $r=0.1$, $\omega_0=0.01$, $e^{-\pi\gamma}=3.69\times10^{-44}$. Thus the exact problem is already extremely close to a direct sum of equal-length sectors at the reference bandwidth.

\section{Rearrangement structure and the limits of shell compression}
For a fixed translation length $L$, let $S\subset\ZZ$ be the occupied midpoint indices and $U_S=\bigcup_{c\in S}J_c$. The equal-length contribution is
\begin{equation}
D_L(S)=W_{LL}\,\mathbb P(M_A\in U_S,M_B\in U_S).
\end{equation}
Because the bivariate Gaussian density is compatible with the Brascamp--Lieb--Luttinger rearrangement inequality \cite{KarlinRinott1980,RinottScarsini2006,BrascampLiebLuttinger1974}, the following exact reduction holds.

\begin{theorem}[Odd-cardinality midpoint compression]\label{thm:compression}
Fix $|S|=2h+1$. Among measurable midpoint sets of total Lebesgue measure $(2h+1)a$, the probability $\mathbb P(M_A\in U,M_B\in U)$ is maximized by the centred interval. Among unions of half-period cells this is
\begin{equation}
U_h^\star=\bigcup_{c=-h}^{h}J_c.
\end{equation}
Thus a fixed odd-cardinality equal-length cohort is optimally placed on consecutive midpoint cells centred at the origin.
\end{theorem}

The statement does not imply that every unrestricted matching is a partition into centred shells. Different length classes compete for cells. In a symmetric $2N$-cell window, a convenient shell grammar leads to a one-dimensional dynamic program and an exponentially small error within that grammar; the precise recurrence and error estimate are recorded in Appendix~\ref{app:shells}.

More importantly, the stronger universal shell-compression conjecture is false.

\begin{proposition}[Finite counterexample]\label{prop:counterexample}
For the $12$-cell translation problem ($N=6$) at $r=0.1$, complete enumeration of all $6!=720$ opposite-parity matchings exhibits, as $\omega_0$ is decreased, a centred-shell phase, then a symmetry-related shifted-shell doublet, and finally a two-block length-$3$ translation crystal. Hence a translation matching outside the centred-shell grammar is globally optimal on a nonempty finite-band interval.
\end{proposition}

This counterexample sharpens the mathematical interpretation of shellization. The rearrangement theorem is exact for one length class; the full translation problem is a coupled packing problem whose narrow-band limit allows displaced and periodic structures. The finite-band phase data therefore describe genuine optima of the stated finite windows, not a universal theorem on all measurable translation matchings.

\begin{figure}[t]
\centering
\begin{tikzpicture}[x=0.42cm,y=0.38cm,>=Latex,thick]
\foreach \x in {-4,-3,-2,-1,0,1,2,3}{
  \fill (\x,0) circle (1.4pt);
}
\draw[gray!70] (-4.4,0)--(3.4,0);
\draw[blue!60!black,bend left=35] (-1,0) to (0,0);
\draw[blue!60!black,bend left=35] (-2,0) to (1,0);
\draw[blue!60!black,bend left=35] (-3,0) to (2,0);
\node[below] at (-0.5,-0.35) {equal-length cohort};
\begin{scope}[yshift=-2.3cm]
\foreach \x in {-4,-3,-2,-1,0,1,2,3}{\fill (\x,0) circle (1.4pt);}
\draw[gray!70] (-4.4,0)--(3.4,0);
\draw[red!65!black,bend left=26] (-4,0) to (-1,0);
\draw[red!65!black,bend left=26] (-3,0) to (0,0);
\draw[red!65!black,bend left=26] (1,0) to (3,0);
\node[below] at (-0.5,-0.35) {shifted / block structure};
\end{scope}
\end{tikzpicture}
\caption{Schematic distinction between a consecutive equal-length shell and a displaced block geometry. Proposition~\ref{prop:counterexample} shows that the latter cannot be excluded by a universal compression principle.}
\label{fig:shells}
\end{figure}

\section{Reflection branch}
The second primitive family in Theorem~\ref{thm:classification} is analytically simpler. For $k\in\ZZ$ define
\begin{equation}
R_k(q)=-q+(2k+1)a.
\end{equation}
Each $R_k$ is a measure-preserving involution satisfying the fibre condition and is therefore wedge-local by Theorem~\ref{thm:wedge}.

\begin{theorem}[Global affine-reflection correlation]\label{thm:reflection}
For $R_k$ the exact cross-wedge correlation is
\begin{equation}\label{eq:reflection}
C_R(k)=\frac{1}{D}\exp\!\left[-\frac{\pi\delta(2k+1)^2}{4D}\right],
\end{equation}
with $D=\nu_A\nu_B-\kappa^2$ and $\delta=\nu_A+\nu_B-2\kappa$. Consequently $k=0$ or $k=-1$ maximizes the affine-reflection family.
\end{theorem}

\begin{proof}
For $R_k$, the midpoint $m=(q+R_kq)/2$ is fixed at $(2k+1)a/2$, while the displacement $d=R_kq-q$ spans the real line. In midpoint-displacement variables the Gaussian kernel factorizes into the centred midpoint density with covariance $Q$ and the displacement quadratic form $P$. Since $\det Q=\det P=D/4$, the normalization is $D^{-1}$, and evaluation at the fixed midpoint gives \eqref{eq:reflection}.
\end{proof}

At $(r,\omega_0)=(0.1,0.01)$,
\begin{equation}
C_R=0.96917057543459\ldots .
\end{equation}
The exact reflection Bell threshold at $r=0.1$ is
\begin{equation}
\omega_{0,R}=0.02203669210198\ldots .
\end{equation}
In the concentration limit,
\begin{equation}
B_R\longrightarrow \sqrt2\left(\frac12+\frac{1}{1+(2/3)r^2}\right),
\end{equation}
so the affine-reflection branch remains asymptotically Bell-violating for $r<0.37517260299\ldots$.

Translations and reflections may also be mixed cellwise. In a $12$-cell window the complete mixed grammar contains $140{,}152$ involutions. Exhaustive evaluation with the same adjacent exterior convention gives
\begin{equation}
C_{\Gamma,12}^{\rm mix}=0.96926709992377\ldots,
\end{equation}
and the corresponding CHSH value quoted below. This is an exact finite-window result within the declared mixed grammar, not a global theorem over arbitrary measurable branch mixtures.

\section{Bell consequence}
Let
\begin{equation}
X=\cos A,\qquad Y=\cos C
\end{equation}
be the cosine axes in the complementary wedges, and let $\Gamma_A,\Gamma_B$ be matched exact Clifford involutions. Choose
\begin{equation}
A_0=X,\qquad A_1=\Gamma_A,
\end{equation}
\begin{equation}
B_0=\frac{Y+\Gamma_B}{\sqrt2},\qquad
B_1=\frac{Y-\Gamma_B}{\sqrt2}.
\end{equation}
Because $\{Y,\Gamma_B\}=0$ and $Y^2\le I$, $B_0$ and $B_1$ are self-adjoint contractions. The CHSH operator reduces to
\begin{equation}
\mathcal C=\sqrt2\,(XY+\Gamma_A\Gamma_B).
\end{equation}
The cosine correlation is
\begin{equation}\label{eq:cx}
C_X=\frac12\left[
 e^{-\frac\pi4(\nu_A+\nu_B-2\kappa)}+
 e^{-\frac\pi4(\nu_A+\nu_B+2\kappa)}
\right],
\end{equation}
and hence
\begin{equation}\label{eq:bell}
B=\sqrt2\,(C_X+C_\Gamma).
\end{equation}
The use of self-adjoint contractions does not alter the classical CHSH bound: every such operator $S$ is the first moment of the binary POVM $E_\pm=(I\pm S)/2$.

\begin{proposition}[Analytic finite-band violation]\label{prop:analyticbell}
For the affine reflection $R(q)=\sqrt\pi-q$, equations \eqref{eq:reflection} and \eqref{eq:cx} give, at $(r,\omega_0)=(0.1,0.01)$,
\begin{equation}
B_R=2.0603750381652\ldots>2.
\end{equation}
The violation therefore follows from closed Gaussian formulas and does not depend on a finite-window optimization. Since $B_R$ is continuous on the parameter region $D>0$, it persists on an open neighbourhood of this point. In the concentration limit,
\begin{equation}
B_R\longrightarrow \sqrt2\left(\frac12+\frac{1}{1+(2/3)r^2}\right),
\end{equation}
which exceeds $2$ for $r<0.37517260299\ldots$.
\end{proposition}

Translation-cell involutions furnish an independent, rigorously tail-controlled violation; the finite-window values, threshold data and tail estimate are collected in Appendix~\ref{app:numerics}. Those computations support the geometric analysis of the translation sector but are not needed for Proposition~\ref{prop:analyticbell}.

\section{Relation to continuous-variable pseudospins and recent QFT constructions}
Continuous-variable pseudospin operators are established constructions \cite{Chen2002PRL,ChenZhang2002,Larsson2004}. In particular, the interval projection and box-translation operators in \cite{Larsson2004} already provide an exact Pauli algebra on $L^2(\RR)$. The present translation cells contain that construction as a special case. The claim here is narrower and QFT-specific: a modularly localized Bose canonical pair turns these fibre permutations into elements of a wedge algebra, and the nonzero-bandwidth two-wedge correlation reduces exactly to the Gaussian functional \eqref{eq:qap}. The reflection branch then appears on the same footing as translations in the complete pointwise solution of the cosine fibre equation.

This also distinguishes the present aim from recent constructions optimized for large Bell values. Finite Weyl polynomials combine noncommuting local axes and can approach $2\sqrt2$ \cite{Caribe2026FiniteWeyl}; chiral-boson vertex operators give an explicit saturating realization \cite{Caribe2026Vertex}. The number $B_R$ in Proposition~\ref{prop:analyticbell} is therefore not presented as a record violation. Its role is to show that the fibre-involution framework gives a closed analytic violation in the Gaussian vacuum while retaining an exact local operator-algebraic description.

The locality statement should also be distinguished from detector implementation. Theorem~\ref{thm:wedge} places the observables in the wedge algebra; it does not compile them into a concrete probe protocol. General frameworks for local QFT measurements exist \cite{FewsterVerch2020,Bostelmann2021,FewsterJubbRuep2023}, but that operational question is separate.

\section{Scope and open mathematical questions}
Three limitations delimit the result. First, the type-I realization is tied to the nondegenerate finite-dimensional CCR plane generated by the chosen free-field observables. We do not assert an analogous canonical construction for a general interacting local algebra. Second, the global optimization problem over all measurable mixtures of translation and reflection branches remains open. The finite translation and mixed enumerations in the appendices are exact only within their declared grammars and windows. Third, the narrow-band counterexample of Proposition~\ref{prop:counterexample} rules out a universal centred-shell theorem: any complete asymptotic classification must allow displaced equal-length runs and periodic block structures.

The most natural mathematical problem is therefore not to increase a single Bell value but to classify the global matching geometry of the positive kernel \eqref{eq:qap}. The exponential suppression of unequal lengths suggests an effective decomposition into weakly coupled length sectors, whereas the counterexample shows that competition for cells prevents a naive independent rearrangement of those sectors. A useful continuum theory would have to control the number, length, displacement and boundary cost of repeated blocks as $\omega_0\downarrow0$.

A second direction is to enlarge the fibre principle beyond the cosine model while preserving explicit local correlations. Theorem~\ref{thm:wedge} applies to any real bounded $F(q)$ once a measure-preserving involution satisfying $F\circ T=-F$ is available. Determining which choices of $F$ lead simultaneously to a tractable fibre geometry and to closed quasifree correlations would give a broader class of exactly analyzable local Bell observables.

\section{Conclusion}
The spectral-fibre formulation converts an exact local anticommutation problem into a measurable matching problem. For the cosine axis the pointwise solutions are exhausted by odd half-period translations and odd-centred reflections. A nondegenerate wedge-local CCR plane then turns every measure-preserving involutive matching into a bounded Hermitian unitary in the wedge algebra. This localization step is a direct consequence of standard CCR representation theory, but it provides the bridge needed to use continuous-variable fibre permutations as relativistic local observables.

At finite modular bandwidth, translation cells have the exact positive Gaussian functional \eqref{eq:qap}. The equal-length part obeys a rearrangement theorem, while the full matching problem does not reduce universally to centred shells. The affine reflection has the closed correlation \eqref{eq:reflection} and gives the analytic finite-band Bell violation of Proposition~\ref{prop:analyticbell}. These statements are independent of any claim of global optimization or record Bell violation.

The resulting picture is therefore structural: local CCR subfactors support explicit anticommuting spectral-fibre partners, and their finite-band correlations can be reduced to exact Gaussian expressions. The unresolved problem is the global combinatorics of the admissible fibre matchings, not the locality of the constructed operators.

\appendix
\section{Mismatch coefficient and off-length suppression}\label{app:mismatch}
Let
\begin{equation}
M_-=\begin{pmatrix}\nu_A&-\kappa\\-\kappa&\nu_B\end{pmatrix},\qquad D=\det M_-.
\end{equation}
For $a=(1,-1)^T$, minimizing $x^TM_-x$ subject to $a^Tx=1$ yields
\begin{equation}
\min x^TM_-x=\frac{1}{a^TM_-^{-1}a}.
\end{equation}
Since
\begin{equation}
M_-^{-1}=\frac1D\begin{pmatrix}\nu_B&\kappa\\\kappa&\nu_A\end{pmatrix},
\end{equation}
we obtain $a^TM_-^{-1}a=\delta/D$, and hence the minimum $D/\delta=\gamma$. The same calculation with $M_+=(\begin{smallmatrix}\nu_A&\kappa\\\kappa&\nu_B\end{smallmatrix})$ and $a=(1,1)^T$ proves the corresponding plus-sign inequality.

For distinct positive odd integers $\ell\ne m$, $|\ell-m|\ge2$ and $\ell+m\ge2$. Applying the two quadratic inequalities to the exponentials defining $W_{\ell m}$ gives
\begin{equation}
W_{\ell m}\le4e^{-\pi\gamma}.
\end{equation}
For a matching with $N$ transpositions the complete unequal-length remainder therefore satisfies
\begin{equation}
0\le R_\sigma\le 4N^2e^{-\pi\gamma}.
\end{equation}
At $r=0.1$, the representative values are
\begin{center}
\small
\begin{tabular}{@{}ccc@{}}
\toprule
$\omega_0$ & $\gamma$ & $e^{-\pi\gamma}$\\
\midrule
$2\times10^{-2}$ & $15.9205$ & $1.90\times10^{-22}$\\
$10^{-2}$ & $31.8332$ & $3.69\times10^{-44}$\\
$2\times10^{-3}$ & $159.154$ & $7.15\times10^{-218}$\\
\bottomrule
\end{tabular}
\end{center}

\section{Shell grammar, finite counterexample and phase data}\label{app:shells}
In the symmetric window
\begin{equation}
\mathcal W_N=\{-N,-N+1,\ldots,N-1\},
\end{equation}
the rainbow pairs $(-j-1,j)$ have lengths $1,3,\ldots,2N-1$ and common midpoint zero. A consecutive odd block $j=s,\ldots,s+2h$ uses the same cells as the equal-length shell
\begin{equation}
L=2(s+h)+1,
\qquad c=-h,\ldots,h.
\end{equation}
This is a cell-preserving re-pairing. If $B_{s,h}$ denotes the exact same-length shell score plus its interaction with a fixed adjacent exterior, and $T_s$ the corresponding adjacent-tail score, the shell grammar obeys
\begin{equation}
F(s)=\max\left\{\max_{h\ge0}\bigl[B_{s,h}+F(s+2h+1)\bigr],T_s\right\},
\end{equation}
where $T_s$ is available only when $N-s$ is even. Omitting interactions between distinct interior shell lengths changes the exact score by at most
\begin{equation}
4N^2e^{-\pi\gamma}.
\end{equation}
At $N=10$, $r=0.1$, $\omega_0=0.01$ this bound is $1.48\times10^{-41}$, while the gap between the two leading shell partitions is $5.07\times10^{-5}$.

The universal centred-shell conjecture fails already at $N=6$. Three representative configurations in midpoint-length notation are
\begin{align}
\mathcal S_6&=\{(-2,5),(-1,5),(0,5),(1,5),(2,5),(0,11)\},\\
\mathcal D^-_6&=\{(-3,5),(-2,5),(-1,5),(0,5),(1,5),(5,1)\},\\
\mathcal C_6&=\{(-4,3),(-3,3),(-2,3),(2,3),(3,3),(4,3)\}.
\end{align}
Together with the reflected partner of $\mathcal D^-_6$, these realize the centred shell $\to$ shifted doublet $\to$ length-$3$ crystal sequence as $\omega_0$ decreases.

For the $20$-cell problem at $r=0.1$, representative exact translation phases are
\begin{center}
\small
\begin{tabular}{@{}llc@{}}
\toprule
$\omega_0$ & central pattern & $B$\\
\midrule
$0.0100$ & $(3^3)|7|9|11|13|15|17|19$ & $2.02505064$\\
$0.0028$ & $(3^3)|(9^3)|13|15|17|19$ & $2.07972991$\\
$0.0021$ & $(3^3)|(9^3)|13|15|17|19$ & $2.08739762$\\
$0.0014$ & $(5^5)|(13^3)|1_{\rm tails}^2$ & $2.09537733$\\
\bottomrule
\end{tabular}
\end{center}
These values are finite-window translation results with the fixed adjacent exterior convention; they do not imply global optimality over all measurable fibre involutions.

\section{Numerical evaluation and tail control}\label{app:numerics}
The finite-window values below are obtained by deterministic evaluation of the exact Gaussian formulas. No Monte Carlo estimator is required. Bivariate cell probabilities are rectangle differences of the two-dimensional normal distribution; candidate matchings are evaluated with the same precomputed Gaussian blocks.

For orientation, at $(r,\omega_0)=(0.1,0.01)$ one finds
\begin{center}
\small
\begin{tabular}{@{}lll@{}}
\toprule
construction & $C_\Gamma$ & $B$ \\
\midrule
translation, $20$ cells & $0.944192454$ & $2.025050640$ \\
global reflection & $0.969170575$ & $2.060375038$ \\
mixed, $12$ cells & $0.969267100$ & $2.060511544$ \\
\bottomrule
\end{tabular}
\end{center}
The reflection entry is analytic and untruncated; the translation and mixed entries are finite-window results with the conventions stated below.

For the $20$-cell translation window, the numbers of opposite-parity perfect matchings are $N!$ with $N=10$, hence $3{,}628{,}800$ candidates. At $(r,\omega_0)=(0.1,0.01)$ the successive completed-window maxima are
\begin{center}
\small
\begin{tabular}{@{}rcc@{}}
\toprule
cells & $C_\Gamma^{\max}$ & $B$\\
\midrule
4 & 0.921914706 & 1.993545146\\
8 & 0.939522080 & 2.018445733\\
12 & 0.943423705 & 2.023963465\\
16 & 0.944141803 & 2.024979008\\
20 & 0.944192454 & 2.025050640\\
\bottomrule
\end{tabular}
\end{center}
The $20$-cell optimizer contains the three midpoint-length pairs $(-1,3),(0,3),(1,3)$ and the singleton rainbow lengths $7,9,11,13,15,17,19$ at midpoint zero.

The infinite exterior can be treated rigorously because every term of \eqref{eq:qap} is nonnegative. Retaining the ten optimized central transpositions and six adjacent exterior blocks at $c=\pm11,\pm13,\pm15$ gives
\begin{align}
C_X&=0.48773458583944921048\ldots,\\
C_{\Gamma,15}&=0.94419245361049243679\ldots,\\
B_{15}&=2.02505063951886138635\ldots .
\end{align}
Hence $B_\infty\ge B_{15}>2$. For the first omitted midpoint, the standardized Gaussian coordinates are
\begin{equation}
z_A=7.32916783196\ldots,\qquad z_B=7.32993487498\ldots .
\end{equation}
The two-sided Mills inequality
\begin{equation}
2[1-\Phi(z)]\le\sqrt{\frac2\pi}\frac{e^{-z^2/2}}{z}
\end{equation}
gives marginal tail bounds below $2.358\times10^{-13}$ and $2.345\times10^{-13}$. Combining these with the off-length estimate gives
\begin{equation}
0\le B_\infty-B_{15}<1.30\times10^{-12}.
\end{equation}

For the reflection branch, direct evaluation of \eqref{eq:reflection} gives $C_R=0.9691705754346\ldots$ and $B_R=2.0603750381652\ldots$ at the same point. At $r=0.1$ the corresponding Bell threshold is $\omega_{0,R}=0.02203669210198\ldots$. The reflection-inclusive $12$-cell enumeration gives $C_{\Gamma,12}^{\rm mix}=0.9692670999238\ldots$ and $B_{12}^{\rm mix}=2.0605115444068\ldots$. Source code reproducing the finite enumerations and an archived snapshot are available at \url{https://github.com/tonynewton79-web/modular-bose-clifford-reproducibility} and DOI \href{https://doi.org/10.5281/zenodo.22257497}{10.5281/zenodo.22257497}.

\section{Additional finite-band geometry}
For completeness, we record two structural diagnostics that are useful for larger finite windows but are not needed for the main Bell conclusion. If $M_A-M_B$ has variance $\delta/2$, then midpoint cells obey
\begin{equation}
P_{cd}\le \min\left\{1,2\exp\left[-\frac{\pi}{\delta}(|c-d|-1)_+^2\right]\right\}.
\end{equation}
Combined with length separation this localizes interactions between defects of candidate matchings. In a target shell, joining adjacent midpoint cells within each equal-length path gives a purely geometric congestion number
\begin{equation}
s_t=1+\sum_\beta(k_\beta-1),
\end{equation}
where $k_\beta$ are the path cardinalities. Every connected component of the target-supported quotient graph on alternating cycles has size at most $s_t$. For the two-shell $20$-cell target at $\omega_0=0.0028$, the repeated $L=3$ and $L=9$ shells each have three midpoint cells, giving $s_t=5$; complete finite enumeration shows that this geometry-only bound is sharp. These graph estimates are finite-window tools and are not used to claim a global optimization theorem.

\balance

\end{document}